\documentclass[final]{IEEEtran}
\usepackage{url}
\usepackage{amssymb,amsmath}
\usepackage[all]{xy}
\usepackage{bm}
\usepackage{subfigure}
\usepackage{psfrag}
\usepackage{xspace}
\usepackage{color,psfrag}
\usepackage{url}
\usepackage{cite,xspace}
\usepackage{algorithm}
\usepackage{algorithmic}
\usepackage{soul}

\usepackage{ifpdf}

\ifpdf
\usepackage[pdftex]{graphicx}
\usepackage{epstopdf}
\else
\usepackage{graphicx}
\fi

\begin{document}

\title{Gaussian Two-way Relay Channel \\ with Private Information for the Relay}

\author{Chin~Keong~Ho,
        Kiran T. Gowda,
         and~Sumei Sun
\thanks{{This work was presented in part at IEEE ICC, Dresden, Germany, June 2009.
}}
\thanks{{C. K. Ho and Sumei Sun are with the Institute for Infocomm Research, A$^\star$STAR, Singapore.
E-mail: \{hock, sunsm\}@i2r.a-star.edu.sg.
}}%
\thanks{{Kiran T. Gowda is with Mobile Communications Department, EURECOM, Sophia-Antipolis, France. He conducted part of this work while with Institute for Infocomm Research, A$^\star$STAR, Singapore.
Email: kiran.gowda@eurecom.fr.
}}
}

\newcommand{\forreport}[1]{}

\newtheorem{conjecture}{Conjecture}
\newtheorem{remark}{Remark}
\newtheorem{insight}{Insight}
\newtheorem{question}{Question}
\newtheorem{proposition}{Proposition}
\newtheorem{corollary}{Corollary}
\newtheorem{lemma}{Lemma}
\newtheorem{assumption}{Assumption}
\newtheorem{theorem}{Theorem}
\newtheorem{example}{Example}
\newtheorem{property}[theorem]{Property}

\newcommand{\myse}{\IEEEyessubnumber} 
\newcommand{\myses}{\myse\IEEEeqnarraynumspace} 

\newcommand{\set}[1]{\mathcal{#1}}

\newcommand{\bn}{\begin{enumerate}}
\newcommand{\en}{\end{enumerate}}

\newcommand{\bi}{\begin{itemize}}
\newcommand{\ei}{\end{itemize}}

\newcommand{\be}{\begin{IEEEeqnarray}{rCl}}
\newcommand{\ee}{\end{IEEEeqnarray}}

\newcommand{\benl}{\begin{IEEEeqnarray*}}
\newcommand{\eenl}{\end{IEEEeqnarray*}}

\newcommand{\bel}{\begin{IEEEeqnarray}}
\newcommand{\eel}{\end{IEEEeqnarray}}

\newcommand{\ben}{\begin{IEEEeqnarray*}{rCl}}
\newcommand{\een}{\end{IEEEeqnarray*}}

\newcommand{\barr}{\begin{array}}
\newcommand{\earr}{\end{array}}

\newenvironment{definition}[1][Definition:]{\begin{trivlist}
\item[\hskip \labelsep {\it #1}]}{\end{trivlist}}

\newcommand{\ud}{\mathrm{d}} 

\newcommand{\FigSize}{0.6}
\newcommand{\FigSizeSmall}{0.5}

\newcommand{\avesnr} {\bar{\gamma}} 
\newcommand{\snr} {\gamma} 

\newcommand{\re}[1]{(\ref{#1})}

\newcommand{\Pe} {P_{\mathrm {e}}} 

\newcommand{\goodgap}{%
\hspace{\subfigtopskip}%
\hspace{\subfigbottomskip}}

\newcommand{\dhat}[1]{\Hat{\Hat{#1}}} 
\newcommand{\that}[1]{\Hat{\Hat{\Hat{#1}}}} 
\newcommand{\dtilde}[1]{\Tilde{\Tilde{#1}}} 
\newcommand{\ttilde}[1]{\Tilde{\Tilde{\Tilde{#1}}}} 

\newcommand{\trace}[1]{\mathrm{tr}\{#1\}}

\renewcommand{\Pe} {P^{(n)}_{\mathrm {e}}} 

\renewcommand{\FigSize}{0.46}

\newcommand{\Rmac}{\bar{R}}

\newcommand{\sou}{\ensuremath{{\mathsf{S}}}\xspace}
\newcommand{\dest}{\ensuremath{{\mathsf{D}}}\xspace}
\newcommand{\relay}{\ensuremath{{\mathsf{R}}}\xspace}

\newcommand{\rlabel}{\mathrm{r}}
\newcommand{\Yr}{Y_{\rlabel m}}
\newcommand{\Yrn}{Y_{\rlabel}^n}
\newcommand{\Yrall}{Y_{\rlabel 1}, \cdots, Y_{\rlabel n}}
\newcommand{\Xr}{X_{\rlabel m}}
\newcommand{\yrn}{y_{\rlabel}^n}
\newcommand{\yr}{y_{\rlabel m}}
\newcommand{\xr}{x_{\rlabel m}}
\newcommand{\Woner}{W_{1\rlabel}}
\newcommand{\Wtwor}{W_{2\rlabel}}
\newcommand{\Whatoner}{\hat{W}_{1\rlabel}}
\newcommand{\Whattwor}{\hat{W}_{2\rlabel}}
\newcommand{\Xoner}{X_{1\rlabel m}}
\newcommand{\Xtwor}{X_{2\rlabel m}}
\newcommand{\Poner}{P_{1\rlabel}}
\newcommand{\Ptwor}{P_{2\rlabel}}
\newcommand{\Wonerhat}{\hat{W}_{1\rlabel}}
\newcommand{\Wtworhat}{\hat{W}_{2\rlabel}}
\newcommand{\wonerhat}{\hat{w}_{1\rlabel}}
\newcommand{\wtworhat}{\hat{w}_{2\rlabel}}
\newcommand{\woner}{w_{1\rlabel}}
\newcommand{\wtwor}{w_{2\rlabel}}
\newcommand{\Wir}{W_{i\rlabel}}
\newcommand{\wir}{w_{i\rlabel}}

\newcommand{\Pir}{P_{i\rlabel}}
\newcommand{\Prone}{P_{\rlabel 1}}
\newcommand{\Prtwo}{P_{\rlabel 2}}
\newcommand{\Pronesub}{P_{\rlabel 1,k}}
\newcommand{\Prtwosub}{P_{\rlabel 2,k}}
\newcommand{\Crmin}{C_{\mathrm{bc}}}
\newcommand{\Crminsub}{C_{\mathrm{bc},k}}
\newcommand{\Crminpermsub}{C_{\mathrm{bc},\pi(k)}}
\newcommand{\Pri}{P_{\rlabel i}}

\newcommand{\Cmac}{\mathcal{C}_{\mathrm{ma}}}
\newcommand{\Ccomp}{\mathcal{C}_{\mathrm{ma}}}
\newcommand{\Rcomp}{\mathcal{R}_{\mathrm{ma}}}
\newcommand{\Cout}{\overline{\mathcal{C}}_{\mathrm{ma}}}
\newcommand{\Cbc}{\mathcal{C}_{\mathrm{bc}}}
\newcommand{\OBbc}{\overline{\mathcal{C}}_{\mathrm{bc}}}
\newcommand{\Rbc}{\mathcal{R}_{\mathrm{bc}}}
\newcommand{\Cbcsub}{\mathcal{C}_{\mathrm{bc},k}}
\newcommand{\Rsep}{\mathcal{R}_{\mathrm{1}}}
\newcommand{\Rsepma}{\mathcal{R}_{\mathrm{1,ma}}}
\newcommand{\Rts}{\mathcal{R}_{\mathrm{2}}}
\newcommand{\Rtsma}{\mathcal{R}_{\mathrm{2,ma}}}
\newcommand{\Rtssub}{\mathcal{R}_{\mathrm{ts},k}}
\newcommand{\Rgeneral}{\mathcal{R}_{\mathrm{gen}}}

\newcommand{\rtwouser}{\widetilde{\mathbf{r}}}
\newcommand{\rallsub}{\overline{\mathbf{r}}}

\newcommand{\amax}{\alpha^{\mathrm{max}}}

\newcommand{\Nsub}{K}

\newcommand{\Rmin}{R^{\mathrm{min}}}

\newcommand{\Roner}{R_{1\rlabel}}
\newcommand{\Rtwor}{R_{2\rlabel}}
\newcommand{\Rir}{R_{i\rlabel}}

\maketitle



\begin{abstract}
We introduce a generalized  two-way relay  channel where two sources exchange information (not necessarily of the same rate) with help from a relay, and each source additionally sends private information to the relay. We consider the Gaussian setting where all point-to-point links are Gaussian channels. For this channel, we consider a two-phase protocol consisting of a multiple access channel (MAC) phase and a broadcast channel (BC) phase. We propose a general decode-and-forward (DF) scheme where the MAC phase is related to computation over MAC, while the BC phase is related to BC with receiver side information. In the MAC phase, we time share a capacity-achieving code for the MAC and a superposition code with a lattice code as its component code. We show that the proposed DF scheme is near optimal for any channel conditions, in that it achieves rates within half bit of the capacity region of the two-phase protocol.
\end{abstract}

\section{Introduction}

Two-way relaying  is an effective means of exchanging information between two sources $\sou_1,\sou_2$ with help from a relay $\relay$ \cite{KimMitranTarokhIT08, Hammerstromspawc07,Schnurr07,Wilson10,OechteringBoche07,Rankov06,hoICC08,NamChungLee10}.
While more phases can be used \cite{KimMitranTarokhIT08},  a two-phase protocol that is relevant when the sources cannot listen to each other is typically considered: the relay listens in the first phase, then performs relaying in the second phase.
Two relaying schemes are widely studied.
In the decode-and-forward (DF) scheme \cite{Rankov06, KimMitranTarokhIT08,Hammerstromspawc07,Schnurr07,Wilson10,OechteringBoche07}, the relay first decodes some or all of the information bits from both sources, while in the amplify-and-forward scheme \cite{Rankov06, hoICC08}, the relay simply forwards the received symbols.
In both schemes, each source removes the self-interference that originates from itself in the first phase, so as to decode the desired message in the second phase.

In the current literature for two-way relaying, e.g., \cite{KimMitranTarokhIT08, Hammerstromspawc07,Schnurr07,Wilson10,OechteringBoche07,Rankov06,hoICC08,NamChungLee10}, the relay does not recover any information from the sources explicitly for its own use. In practice, the relay may require side information from the sources to facilitate two-way relaying,
e.g., to achieve phase or frequency synchronization, or to update channel state information or queue information.
For simplicity, we model the required side information as {\em private messages} $\Woner,\Wtwor$ to be communicated from sources $\sou_1,\sou_2$, respectively, to the relay, see Fig.~\ref{fig:overview}. In general, common (i.e., non-private) information may also be sent to all the other nodes, e.g., in \cite{OechteringBoche07} the relay sends a common message to both sources, but such a multicast scenario is not covered here.

In this paper, we consider a generalized two-way relay channel where two sources exchange information (not necessarily of the same rate) and each source sends private information to the relay.
We consider the Gaussian setting where all point-to-point channels are Gaussian channels. We focus on the two-phase protocol shown in Fig.~\ref{fig:protocol1}. In the multiple access channel (MAC) phase, the sources transmit, while in the broadcast channel (BC) phase, the relay transmits.
Both phases are carried out over orthogonal radio resources.
This protocol is  relevant if the communication link between $\sou_1$ and $\sou_2$ is weak or absent.

\begin{figure}
\centering
    \hspace{0cm}
    \xymatrix@R=0.1cm@C=2cm{
    *+[o][F-]{\sou_1}\ar@/^2ex/@{-->}^{W_{12}}[dd]\ar@/^2ex/@{-->}^{\Woner}[rd] \\
    & *++[o][F-]{\relay}   \\
    *+[o][F-]{\sou_2}\ar@/^2ex/@{-->}[uu]^{W_{21}}\ar@/_2ex/@{-->}[ru]_{\Wtwor} \\ \;
    }
\caption{A generalized two-way relay. Each dotted arrow represents the flow of information, via a message $W_{ij}$, from source $\sou_i$ to its final destination, namely, source $\sou_j$ or relay $\relay$.
}
\label{fig:overview}
\end{figure}
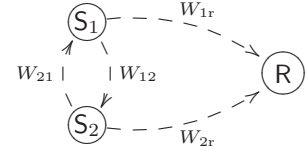

\begin{figure}
\centering
\subfigure[Multiple access channel (MAC) phase.
] 
{
    \label{subfig:mac}
    \xymatrix@R=0.1cm@C=0.5cm{
    \Woner, W_{12} \ar@{=>}[r] & *+[o][F-]{\sou_1}\ar@{->}[rdr]_>>>>{\Yrn}_<<<{X_1^n}
    \\
    & &  & *++[o][F-]{\relay}\ar@{=>}[r] & \Wonerhat, \Wtworhat\hspace{2cm}  
    \\
    W_{21}, \Wtwor \ar@{=>}[r] & *+[o][F-]{\sou_2}\ar@{->}[urr]^<<<{X_2^n} \\ \;
    }
}
\subfigure[Broadcast (BC) phase.
] 
{
    \label{subfig:br}
    \xymatrix@R=0.1cm@C=0.5cm{
    {
    \hat{W}}_{21} \ar@{<=}[r] & *+[o][F-]{\sou_1}\ar@{<-}_<<<<{Y_1^n}_>>>{X_{\rlabel}^n}[rrd]&
    \\
    && & *++[o][F-]{\relay}\ar@{<=}[r] & \Yrn \hspace{1.8cm}  
    \\
    {\hat{W}}_{12} \ar@{<=}[r] & *+[o][F-]{\sou_2}\ar@{<-}^<<<<{Y_2^n}[urr]&  \\ \;
    }    
}
\caption{A two-phase protocol for the generalized two-way relaying.
A transmission is represented by $\rightarrow$, an encoding or decoding operation by $\Rightarrow$.
}
\label{fig:protocol1}
\end{figure}
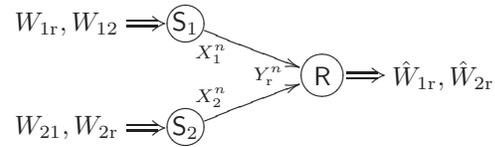
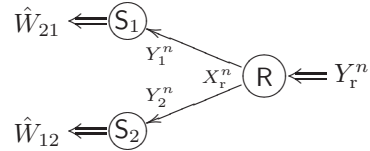

We view the generalized two-way relay channel as an amalgam of a conventional two-way relay channel where no private information is sent, and a conventional MAC where only private information is sent by the sources to the relay.
With this view, we propose a DF scheme for the two-phase protocol.
This DF scheme corresponds closely to computation over MAC \cite{NazerGastpar07, NazerGastparAllert07} in the MAC phase, and to the BC with receiver side information \cite{WuISIT07,TuncelIT07,Oechtering08} in the BC phase.
Specifically, in the MAC phase we propose an equal-exchange-rate with bit relabeling (EER-BR) scheme that involves two steps. First, some of the exchange message bits are relabeled as private information bits such that the messages to be exchanged are of equal rates. Second, to transmit the relabeled messages, we time share two coding schemes, namely a capacity-achieving code for the conventional MAC and a superposition code with a lattice code as its component code.
The overall DF scheme is near optimal in that reliable decoding is possible if $(R_{12},R_{21},\Roner,\Rtwor)$ lies within half bit of the capacity region of the two-phase protocol, where $R_{ij}$ is the achievable rate of the message from node $i$ to node $j$.
This holds for arbitrary transmission powers and channel conditions in the MAC and BC phases, e.g., channel reciprocity may not hold in general.
Our result may be treated as a generalization of the result in \cite{Wilson10}, where the conventional two-way relay with $\Roner=\Rtwor=0$ and $R_{12}=R_{21}$ is considered. 
Key to our DF scheme is the lattice code used for computation over MAC that is introduced in \cite{Wilson10,NazerGastparAllert07}.

{\em Notations:}
Let $C(x)=1/2 \log(1+x), D(x)=1/2 \max\{0,\log(1/2+x)\}, x\geq 0$.
Logarithms are of base two. Rates are expressed in bit/symbol. 
Upper case letters denote random variables. Lower case letters denote the values of random variables.
We collect $n$ elements $X_{1}, \cdots, X_{n}$ as a vector $X^n$.

\section{System Model}\label{sec:system}

The generalized two-way relay channel is shown in Fig.~\ref{fig:overview}. Two sources $\sou_1, \sou_2$ exchange messages $W_{12} \in\{1,\cdots,2^{nR_{12}}\}$ and $W_{21}\in\{1,\cdots,2^{nR_{21}}\}$, respectively.
In addition, $\sou_i$ sends a message $\Wir\in\{1,\cdots,2^{n\Rir}\}$ to the relay, $i=1,2$.
The messages are generated independently with a uniform distribution.

\subsection{Two-Phase Protocol}

We consider the two-phase protocol as shown in Fig.~\ref{fig:protocol1}, which consists of a MAC phase and a BC phase.
In each phase, $n$ channel symbols are transmitted; the extension for different number of channel symbols in both phases is straightforward.
The discrete time index $m$ ranges from $1$ to $n$ in both phases.
\\
{\underline{MAC phase:}}
$\sou_1$ encodes both messages $W_{12}, \Woner$ to form the codeword $X_1^n$ for transmission in the MAC phase. Similarly, $\sou_2$ encodes  $W_{21}, \Wtwor$ to form the codeword $X_2^n$.
The relay thus receives at time $m$
\begin{IEEEeqnarray}{rCl}\label{eqn:Yma}
\Yr=  X_{1m}(W_{12}, \Woner) + X_{2m}(W_{21}, \Wtwor) + Z_m, 
\end{IEEEeqnarray}
where
$Z_m\sim\mathcal{N}(0,1)$ is zero-mean unit-variance i.i.d. Gaussian noise.
All signals are real-valued.
We impose the power constraints $\sum_{m=1}^n |x_{im}|^2 \leq n P_i, i=1,2.$
Without loss of generality, let
$
P_1  \leq P_2.
$
\\
{\underline{BC phase:}}
The relay uses the received signal $\Yrn$ to decode for its private messages as $\Wonerhat, \Wtworhat$, and also to form a codeword $X_{\rlabel}^n$ for transmission in the BC phase. Source $\sou_i, i=1,2,$ thus receives at time $m$
\begin{IEEEeqnarray}{rCl}\label{eqn:Ybc}
 Y_{im} = \sqrt{\Pri} \Xr(\Yrn) +  Z'_{im}, 
\end{IEEEeqnarray}
where $Z'_{im}\sim\mathcal{N}(0,1)$ is i.i.d. Gaussian noise and $\Pri$ is the SNR from the relay to source $\sou_i$.
Without loss of generality, we impose the power constraint $\sum_{m=1}^n |x_{\mathrm{r}m}|^2 /n  \leq 1$.
Using $Y_1^n$, as well as the previously transmitted messages $W_{12}, \Woner$ as side information,  $\sou_1$ decodes its desired message as $\hat{W}_{21}$. Note that $X_{1}^n$ can be constructed from $W_{12}, \Woner$ by the source $\sou_1$ (during decoding) and hence is also implicitly available as side information.
Similarly, using $Y_2^n$  and side information $(W_{21}, \Wtwor)$,  $\sou_2$ decodes its desired message as $\hat{W}_{12}$.

An error event is said to occur if at least one of the messages in $W\triangleq (W_{12}, W_{21}, \Woner, \Wtwor)$ is not decoded correctly by the intended {\em final} destination at the end of a protocol cycle.
Thus, it is not necessary for the relay to decode $W_{12}$ or $W_{21}$.
The rate tuple $(R_{12},R_{21}, \Roner, \Rtwor)\in\mathbb{R}^4_+$ is said to be {\em achievable} if the average probability of error $\Pe$ can be driven to zero for $n\rightarrow \infty$.
An {\em achievable rate region} $\mathcal{R}$ is a collection of achievable rate tuples.
The {\em capacity region} $\mathcal{C}$ is the closure of the set of all achievable rate tuples, and its outer bound is denoted as $\overline{\mathcal{C}}$. Thus, $\mathcal{R}\subseteq \mathcal{C} \subseteq \overline{\mathcal{C}}\subseteq\mathbb{R}_{+}^4$.

\subsection{An Outer Bound for the Capacity Region}
Theorem~\ref{thm:converse} states an outer bound $\overline{\mathcal{C}}$ for the two-phase protocol, which holds for any source and relay processing, and for any Gaussian channels in the MAC and BC phases, i.e., $P_1, P_2, \Prone, \Prtwo$ are arbitrary and so channel reciprocity is not assumed.

We recall that all messages are mutually independent.
The  $\sou_1$-and-$\sou_2$-to-$\relay$ channel, as well as the $\relay$-to-$\sou_1$ and $\relay$-to-$\sou_2$ channels, are memoryless with Gaussian transition probabilities given by  $p^*(\yr|x_{1m},x_{2m}), p^*(y_{1m} | \xr),  p^*( y_{2m}| \xr)$, respectively.
In general, we express the encoding functions for $\sou_1$, $\sou_2$ and $\relay$ as  $p(x_{1}^n| \woner, w_{12})$, $p(x_{2}^n| \wtwor, w_{21})$,   $p(x_{\rlabel}^n|\yrn)$, and their decoding functions as $p(\hat{w}_{21}|y_1^n, \woner, w_{12})$, $p(\hat{w}_{12}|y_2^n, \wtwor, w_{21})$, $p(\wonerhat,\wtworhat|\yrn)$, respectively.
Note that each source can use  its previously transmitted messages as side information for decoding.
Thus, the joint distribution factorizes as 
\begin{IEEEeqnarray}{rrl}\nonumber
&& p(w, x^n_1,x^n_2,\yrn, y_1^n, y_2^n, \hat{w})
= p(\woner) p(w_{12}) p(\wtwor) p(w_{21}) \\ \nonumber
&&\times p(x_{1}^n| \woner, w_{12}) p(x_{2}^n| \wtwor, w_{21})   p^*(\yrn|x_{1}^n,x_{2}^n) p(\wonerhat,\wtworhat|\yrn)  \\ \nonumber
&&
\times p(x_{\rlabel}^n|\yrn)
p^*(y_{1}^n | x_{\rlabel}^n)  p^*( y_{2}^n| x_{\rlabel}^n)
 p(\hat{w}_{21}|y_1^n, \woner, w_{12}) \\
&& p(\hat{w}_{12}|y_2^n, \wtwor, w_{21}),
\label{eqn:pdfall}
\end{IEEEeqnarray}
where $w=(\woner, w_{12}, \wtwor, w_{21})$ and
$\hat{w}$ denotes the decoded message of $w$.
Fig.~\ref{fig:converse_general}  relates the random variables by a dependence diagram.

\begin{figure}
\centering
\hspace{0cm}
\xymatrix@R=0.4cm@C=0.35cm{
\Woner\ar@{->}[rd] \ar@{->}@/^4ex/[rrrrd] \\
  W_{12}\ar@{->}[r] \ar@{->}@/^4ex/[rrrr]  & X_1^n\ar@{->}[dr]  && Y_1^n  \ar@{->}[r] & \hat{W}_{21}\\
& & \Yrn \ar@{->}[r] \ar@{->}[d] & X_{\rlabel}^n  \ar@{->}[d] \ar@{->}[u]&\\
W_{21} \ar@{->}[r] \ar@{->}@/_4ex/[rrrr]  & X_2^n\ar@{->}[ur] & (\Wonerhat, \Wtworhat) & Y_2^n  \ar@{->}[r] & \hat{W}_{12}\\
\Wtwor\ar@{->}[ru] \ar@{->}@/_4ex/[rrrru]& &
\;
}
\caption{Dependence diagram for the generalized two-way relay channel.}
\label{fig:converse_general}
\end{figure}
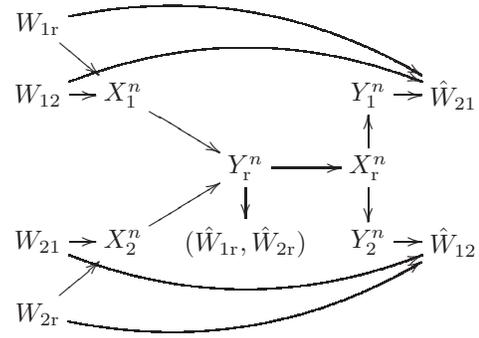


\newcommand{\mybox}[1]{\mbox{\small{#1}}}

\begin{theorem}\label{thm:converse}
Consider the Gaussian two-way relay channel with distribution \re{eqn:pdfall}.
If $\Pe\rightarrow 0$ for $n\rightarrow \infty$, then $(R_{12}, R_{21},\Roner,\Rtwor)\in \overline{\mathcal{C}}$, where the outer bound is
\begin{IEEEeqnarray}{rlr}\nonumber
\overline{\mathcal{C}}=\big\{&(R_{12}, R_{21},\Roner,\Rtwor)\subseteq \mathbb{R}^4_+ :   &
\\ \label{eqn:thm:converse}
&(R_{12}, R_{21})\in \OBbc,(R_{12}, R_{21},\Roner,\Rtwor)\in \Cout &\big\}
\end{IEEEeqnarray}
where
\begin{IEEEeqnarray}{rcl}\nonumber
\OBbc \triangleq \big\{  &&(R_{12}, R_{21})\subseteq \mathbb{R}^2_+ :  \\
&& R_{12} \leq  C(\Prtwo) , \; R_{21} \leq  C(\Prone)  \big\}
\label{eqn:conversea} \\
\nonumber
\Cout \triangleq  \big\{ &&  (R_{12}, R_{21},\Roner,\Rtwor)\subseteq \mathbb{R}^4_+ : \\
\label{eqn:converse13}
\IEEEyessubnumber\label{eqn:converse1}
&& \Roner + R_{12} \leq C(P_1 ),
\\ \IEEEyessubnumber\label{eqn:converse2}
&& \Rtwor + R_{21} \leq  C(P_2 ) \\
\IEEEyessubnumber\label{eqn:converse3}
&& \Roner + \Rtwor + \max\{R_{12}, R_{21}\} \leq C(P_1  + P_2 )
\big \}.
\;\;\;\;
\end{IEEEeqnarray}
\end{theorem}
\begin{proof}
See proof in Appendix~\ref{append:proof_converse_overall}.
\end{proof}

In Theorem~\ref{thm:converse}, we chose the subscripts in $\OBbc$ and $\Cout$ to emphasize that the regions in \eqref{eqn:conversea} and \eqref{eqn:converse13} are relevant only for the MAC and BC phases, respectively, since the power constraint terms therein relates only to their respective phases.

\begin{remark}
If $R_{12}=R_{21}=0$ (the channel degenerates to a classical MAC),  $\overline{\mathcal{C}}$ reduces to the well known MAC capacity region \cite{Cover06}.
If $\Roner=\Rtwor=0$ (the channel degenerates to a conventional two-way relay channel), $\overline{\mathcal{C}}$ reduces to the outer bound in \cite{Wilson10}.
\end{remark}

\section{Coding Schemes for the Two-Phase Protocol}\label{sec:schemes}

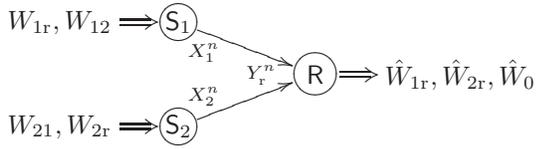
\begin{figure}
\centering
\hspace{0cm}
    \xymatrix@R=0.1cm@C=0.5cm{
    \Woner, W_{12} \ar@{=>}[r] & *+[o][F-]{\sou_1}\ar@{->}[rdr]_>>>>{\Yrn}_<<<{X_1^n}
    \\
    & &  & *++[o][F-]{\relay}\ar@{=>}[r] & \Wonerhat, \Wtworhat, {\hat{W}}_{0} 
    \\
    W_{21}, \Wtwor \ar@{=>}[r] & *+[o][F-]{\sou_2}\ar@{->}[urr]^<<<{X_2^n}  \\ \;
}
\caption{Computation over MAC. 
The relay $\relay$ decodes for $\Woner, \Wtwor$ and a function of messages $W_0= f(W_{12},W_{21})$.}
\label{fig:computeMAC}
\end{figure}

We propose a general DF strategy that relates the MAC and BC phases via an auxiliary message $W_0$.
Using $\Yrn$, the relay decodes for its private information $\Woner,  \Wtwor$, as well as an auxiliary message $W_0$ at rate $R_0$, where $W_0$ is a function of the messages to be exchanged, i.e.,
\forreport{Another possibility is for the relay to compress $W_0$ in a compress-and-forward scheme \cite{Schnurr07}.}
%
\begin{IEEEeqnarray}{rCl}\label{eqn:W0}
W_0=f(W_{12}, W_{21}).
\end{IEEEeqnarray}
Based on the estimate $\hat{W}_0$, the relay then broadcasts a codeword $X_{\rlabel}^n(\hat{W}_0)$ in the BC phase.

This approach in the MAC phase is related to computation over MAC \cite{NazerGastpar07,NazerGastparAllert07}, see
Fig.~\ref{fig:computeMAC}.
For a given function $f$, an error event is said to occur in the MAC phase if at least one of $W_0, \Woner, \Wtwor$ is not decoded correctly. 
The rate tuple $(R_{12}, R_{21}, \Roner,\Rtwor)$ is said to be achievable if the error probability can be driven to zero for $n\rightarrow \infty$.
The rate region in the MAC phase is denoted as $\Rcomp$.

Suppose $\hat{W}_0=W_0$, which occurs with high probability if $(R_{12}, R_{21}, \Roner,\Rtwor)\in\Rcomp$.
Using $Y_1^n$ and the side information $(X_1^n, \Woner, W_{12})$, $\sou_1$ decodes for $W_{21}$.
Similarly, $\sou_2$ decodes for $W_{12}$ using its side information.
This corresponds to a BC with receiver side information \cite{WuISIT07,TuncelIT07,Oechtering08}.
An error event is said to occur if at least one of $W_{12}, W_{21}$ is not decoded correctly.
The rate tuple $(R_{12}, R_{21})$ is said to be achievable if the error probability can be driven to zero for $n\rightarrow \infty$.
The achievable rate region in the BC phase is denoted as $\Rbc$.

Now if $\mathbf{r}=(R_{12}, R_{21}, R_{1r}, R_{2r})$ is achievable for computation over MAC and the same $(R_{12}, R_{21})$ is achievable for BC with receiver side information, then each message $W_{12}, W_{21}, \Woner, \Wtwor$ is decoded correctly by the intended final destination.
Thus, $\mathbf{r}$ is achievable for the two-phase protocol of the generalized two-way relay channel.
An achievable rate region is thus
\begin{IEEEeqnarray}{lLl}\nonumber
\mathcal{R}(\Rbc, \Rcomp) \triangleq \{& (R_{12}, R_{21},\Roner,\Rtwor)\subseteq \mathbb{R}^4_+ :  & \\
& (R_{12}, R_{21})\in \Rbc, & \nonumber \\
& (R_{12}, R_{21},\Roner,\Rtwor)\in \Rcomp & \}.
\label{eqn:rateDF}
\end{IEEEeqnarray}

Next, we consider two specific schemes based on the DF strategy
and quantify their optimality in terms of the achievable rate regions.

\subsection{Conventional MAC Scheme}
In our first scheme, we define $W_0=(W_{12}, W_{21}).$
Thus, the relay decodes for $W=(W_{12},  W_{21}, \Woner, \Wtwor)$ in the MAC phase, then a codeword based on  $W_0$ is transmitted in the BC phase.
We call this the conventional MAC approach, as it can be implemented in the MAC phase using the classical MAC \cite{Cover06}.
Theorem~\ref{thm:1} gives the achievable rate region $\Rsep$. 

\begin{theorem}\label{thm:1}
The achievable rate region of the conventional MAC scheme is $\Rsep=\mathcal{R}(\OBbc, \Rsepma)$,  where
\begin{IEEEeqnarray}{LLL}\nonumber
\Rsepma = \{ & (R_{12}, R_{21}, \Roner,\Rtwor) \subseteq\mathbb{R}_{+}^4 : &  \\
\label{eqn:ratereg1}
\IEEEyessubnumber\label{eqn:ratereg31}
& \Roner+ R_{12} \leq C(P_1 ), &
\\ \IEEEyessubnumber\label{eqn:ratereg32}
& \Rtwor +R_{21}\leq  C(P_2 ), & \\
& \Roner+\Rtwor+R_{12}+R_{21} \leq C(P_1  + P_2 ) & \}.\;\;\;
\IEEEyessubnumber\label{eqn:ratereg33}
\end{IEEEeqnarray}
\end{theorem}
\begin{proof}
In the classical Gaussian MAC described by $\Yr=  X_{1m}(W_{1}) + X_{2m}(W_{2}) + Z_m,$ where $Z_m$ is Gaussian noise and $\sum_{i} |x_{im}|^2\leq n P_i$ for $i=1,2$, a destination decodes messages $W_1,W_2$  at rate $\Rmac_1, \Rmac_2$ respectively. The capacity region is
$\Cmac \triangleq \left\{ (\Rmac_1,\Rmac_2) \subseteq\mathbb{R}_{+}^2:
\sum_{i\in\mathcal{S}} \Rmac_i \leq  C\left(\sum_{i\in\mathcal{S}} P_i \right) \forall \; \mathcal{S}\subseteq \{1,2\}\right\}$ \cite{Cover06}.
In the conventional MAC approach,  the relay becomes the destination with  $W_1=(\Woner, W_{12}), W_2=(\Wtwor, W_{21})$. Substituting $\Rmac_1=\Roner+R_{12}$ and $\Rmac_2=\Rtwor+R_{21}$ into $\Cmac$ then gives $\Rsepma$ in \re{eqn:ratereg1}.

Suppose $\hat{W}_0=W_0$, which occurs with high probability if the rate tuple lies in $\Rsepma$. Then the relay knows all messages in $W$. In this case, the BC capacity with receiver side information is known \cite{WuISIT07,TuncelIT07} and meets the outer bound in the BC phase in Theorem~\ref{thm:converse}, i.e., $\Rbc=\OBbc$.
\end{proof}

\begin{remark}
The conventional MAC scheme is optimal with respect to the BC phase, in the sense that every point in $\OBbc$ can be achieved in the BC phase assuming the messages are always correctly decoded in the MAC phase.  However, comparing the MAC phase region $\Rsepma$ with the corresponding upper bound $\Cout$ shows that the difference of \eqref{eqn:ratereg33} and \eqref{eqn:converse3} can be arbitrarily large at high SNR.
\end{remark}

\subsection{Equal-Exchange-Rate with Bit Relabeling Scheme}\label{sec:ts}

Next, we propose the EER-BR scheme and show that it achieves near-optimal performance.

In this scheme, we use the nested  lattice code $\mathcal{L}$ \cite{ErezZamir04}, associated with a fine lattice $\Lambda_f$  for lattice decoding  and a coarse lattice $\Lambda\subseteq \Lambda_f$ for signal shaping and constraining the power. In \cite{Wilson10}, the lattice code is used for two-way relaying for the case of $R_{12}=R_{21}$ and $\Roner=\Rtwor=0$.
Every lattice codeword $T^n\in \mathcal{L}$ is transmitted over $n$ symbols, and is mapped one-to-one to message $W$ of rate $R_{\mathcal{L}}$ via the mapping $g$ such that $T^n=g(W)$ and $W=g^{-1}(T^n)$.
Define the operation $\oplus$ according to $T^n_{1}\oplus T^n_{2} \triangleq T^n_{1}+ T^n_{2} \bmod\Lambda$ where $T^n_{1}, T^n_{2}\in \mathcal{L}$ and $\bmod\Lambda$ is the modulo operation over $\Lambda$.

\subsubsection{Equal Exchange Rates}

We first consider the EER scheme where we assume $R_{12}=R_{21}=R_0'$. 
We propose time sharing of Schemes 1 and 2 which are described below. 

\underline{Scheme 1 (Conventional MAC):}
Both sources send only their respective private messages $\Woner, \Wtwor$ to the relay.
The messages $W_{12}, W_{21}$ to be exchanged are not sent, i.e., $R_0'=0$.
In Scheme 1, the sources use independent Gaussian codes, which allows any rate pair in the capacity region $\Cmac$ to be achieved in the MAC phase.
Thus, any rate tuple $\mathbf{r}_1=(0,0,\Roner,\Rtwor)$ is achievable for $(\Roner,\Rtwor)\in \Cmac$.

\underline{Scheme 2 (Superposition):}
Recall that $P_1\leq P_2$. 
To send message $W_{12}$, the (weaker) source $\sou_1$ transmits $X_{12}^n$ with power $P_1$.
$\sou_1$ does not transmit any private message, i.e, $\Roner=0$.
To send message $W_{21}$, $\sou_2$  transmits $X_{21}^n$ at the same power of $P_1$.
Moreover, to send its private message $\Wtwor$, source $\sou_2$ employs the superposition technique to transmit $X_{22}^n$  with power $P_2-P_1$.
That is,
\begin{IEEEeqnarray}{rCl}\label{eqn:super1}
X_{1m} &=& \sqrt{P_1}X_{12m}(W_{12}), \\
\label{eqn:super2}
X_{2m} &=& \sqrt{P_1}X_{21m}(W_{21})+\sqrt{P_2-P_1}X_{22m}(\Wtwor),
\end{IEEEeqnarray}
where each codeword is subject to unit power constraints, i.e.,  $\sum_{m=1}^n |x_{12m}|^2 /n \leq 1$,  $\sum_{m=1}^n |x_{22m}|^2/ n \leq 1$, and $\sum_{m=1}^n |x_{21m}|^2 / n\leq 1$.
Here, $X_{22}^n$ is transmitted using a Gaussian code. The remaining signals use the {\em same} lattice code $\mathcal{L}$ of rate $R_{\mathcal{L}}=R_0'$ to give  $X_{21}^n(W_{21})=g(W_{21})=T_{21}^n$ and $X_{12}^n(W_{12})=g(W_{12})=T_{12}^n$.
For decoding, the relay employs successive decoding. Specifically, the relay first decodes for $\Wtwor$ of signal power $P_2-P_1$, treating  $\sqrt{P_1}(T_{12}^n+T_{21}^n)$ as interference of power $2P_1$. The zero-mean Gaussian distribution (with the same interference power) is the worst-case interference distribution, hence the rate $\Rtwor=C((P_2-P_1)/(1+2P_1))$ is achievable.
After reliably decoding $\Wtwor$,  $X_{22m}$ is removed from $X_{2m}$. The received signal after interference cancelation is thus
$\sqrt{P_1}(T_{12}^n+T_{21}^n)+Z^n$.
Then, following the approach in \cite{Wilson10}, the relay decodes for $T^n_0\triangleq T_{12}^n\oplus T_{21}^n$, which  allows $W_0=g^{-1}(T^n_0)$ to be obtained.
The rate $R'_0=D(P_1)$ is achievable by lattice decoding \cite{Wilson10},
thus  $\mathbf{r}_2=(D(P_1), D(P_1), 0, C((P_2-P_1)/(1+2P_1)))$ is achievable for the MAC phase.

\underline{EER Scheme:} 
We time share Schemes 1 and 2 so that  $(1-\alpha) \mathbf{r}_1 + \alpha \mathbf{r}_2$ is achievable for $0\leq \alpha \leq 1$.
The achievable rate region for the EER scheme is then given by
\begin{IEEEeqnarray}{rrl}\nonumber
\Rts = \big\{ && (R_0,R_0,\Roner,\Rtwor) : \\ &&
R_0 = \alpha D(P_1),
(\Roner,\Rtwor)\in \Rts'(\alpha),
0\leq \alpha\leq 1 \big\}
\;\;\;\;\;
\label{eqn:rts}
\end{IEEEeqnarray}
and $\Rts'(\alpha)$ denotes the region of $(\Roner,\Rtwor)$ such that
\begin{IEEEeqnarray}{rCl}\label{eqn:tsall}
\IEEEyessubnumber\label{eqn:ts0}
0\leq \Roner & \leq & (1-\alpha) C(P_1 ), \\
0\leq \Rtwor & \leq & (1-\alpha) C(P_2 ) + \alpha C\left(\frac{P_2-P_1}{1+2P_1}\right) \nonumber \\
&=& C(P_2) - \alpha \Gamma, \IEEEyessubnumber\label{eqn:ts1} \\
\Roner+\Rtwor &\leq&
(1-\alpha) C(P_1  + P_2 )+\alpha C\left(\frac{P_2-P_1}{1+2P_1}\right) \nonumber \\
&=&
C(P_1  + P_2 ) -\alpha C(2P_1)
\IEEEyessubnumber\label{eqn:ts2}
\;\;\;\;\;\;\;
\label{eqn:ts}
\end{IEEEeqnarray}
where $\Gamma\triangleq C(2P_1)+C(P_2)-C(P_1+P_2).$

\newcommand{\tWoner}{\tilde{W}_{1\rlabel}}
\newcommand{\tWtwor}{\tilde{W}_{2\rlabel}}

\subsubsection{Arbitrary Exchange Rates}

Denote the messages  in the EER scheme as $\mathcal{W}\triangleq (W_{12}, W_{21}, \Woner, \Wtwor)$ where $W_{12}, W_{21}$ are at the same rate of $R_0'$.
Denote the  messages  in the EER-BR scheme as $\mathcal{\tilde{W}}\triangleq (\tilde{W}_{12}, \tilde{W}_{21}, \tWoner, \tWtwor)$ where $R_{12}, R_{21}$ can be different.
For arbitrary $R_{12}, R_{21}$, we build on the  EER scheme with the bit-relabeling technique.
The key idea is to use the EER scheme to transmit the exchange messages at a common rate of $R_0'=\min\{R_{12}, R_{21}\}$, and transmit the remaining  $\delta\triangleq |R_{12}-R_{21}|$ bits of the (longer) exchange message  together with the private messages.

First, suppose $R_{12}\leq R_{21}$.
We split the message as $\tilde{W}_{21}=(\tilde{W}_{21}', \tilde{W}_{21}'')$, where $\tilde{W}_{21}'$ and $\tilde{W}_{21}''$ have respective rates $R_{12}$ and $\delta$.
We use the EER scheme by  relabeling the messages as
$W_{12}=\tilde{W}_{12}, W_{21}=\tilde{W}_{21}', \Woner=\tWoner, \Wtwor=(\tWtwor,\tilde{W}_{21}'')$.
That is, $\tilde{W}_{12}, \tilde{W}_{21}'$ become the messages to be exchanged, while $\tilde{W}_{21}''$ is sent as additional ``private" message to the relay (although the relay does not need this message).
Thus, if $(R_{12},R_{12},\Roner,\Rtwor)$ is achievable with the EER scheme, then  $(R_{12},R_{12}+\delta,\Roner,\Rtwor-\delta)$ for $0\leq \delta\leq \Rtwor$ is also achievable with the EER-BR scheme.
From \re{eqn:rts}, the achievable rate region for $R_{12}\leq R_{21}$ is thus
\begin{IEEEeqnarray}{rrl}
\Rtsma = \big\{  && (R_{0},R_{0}+\delta,\Roner,\Rtwor-\delta) : \nonumber \\
&& R_0 = \alpha D(P_1), (\Roner,\Rtwor)\in \Rts'(\alpha),
\nonumber \\ &&
0\leq \alpha\leq 1, 0\leq \delta\leq \Rtwor \big\}.
\label{eqn:rts2}
\IEEEyessubnumber\label{eqn:rts21}
\end{IEEEeqnarray}
Suppose the rate tuple lies in $\Rtsma$. Then $\tilde{W}_0\triangleq (W_0, W_{21}'')$ can be decoded,  where $W_0=g^{-1}(g(W_{12}) \oplus g(W_{21}'))$.
In the BC phase, the relay broadcasts $\tilde{W}_0$ using a Gaussian code. The sources use their side information  to decode their messages. Decoding is reliable if $(R_{12}, R_{21})\in \OBbc$, where the proof follows as a special case of the achievability proof in \cite{WuISIT07} with $R_1=R_2=R_3=0$. In \cite{WuISIT07}, $W_0$ is defined by the bit-wise addition of $W_{12}$ and $W_{21}'$, instead of $W_0=g^{-1}(g(W_{12}) \oplus g(W_{21}'))$ defined here,
but the proof still follows through since each message can always be uniquely mapped to a lattice codeword.

Suppose $R_{21}\leq R_{12}$. Similarly, the achievable rate region in the MAC phase is
\begin{IEEEeqnarray}{rrl}
\Rtsma = \big\{ &&  (R_{0}+\delta,R_{0},\Roner-\delta,\Rtwor) ) : \nonumber \\
&& R_0 = \alpha D(P_1), (\Roner,\Rtwor)\in \Rts'(\alpha), \nonumber
\\ &&
0\leq \alpha\leq 1, 0\leq \delta\leq \Roner \big\}.
\IEEEyessubnumber\label{eqn:rts21}
\addtocounter{equation}{-1}\addtocounter{IEEEsubequation}{1}
\end{IEEEeqnarray}
In the BC phase, decoding is also reliable if $(R_{12}, R_{21})\in \OBbc$.

From the above discussions, we thus obtain Theorem~\ref{thm:rateRts}.
\begin{theorem}\label{thm:rateRts}
The achievable rate region of the EER-BR scheme is $\Rts=\mathcal{R}(\OBbc, \Rtsma)$.
\end{theorem}

\subsubsection{Near Optimality}

The near-optimality of the EER-BR scheme is characterized in Theorem~\ref{thm:Rtscompare}.
First, Lemma~\ref{lem:Rtscompare} establishes the near-optimality of the proposed scheme for the MAC phase.

\begin{lemma}\label{lem:Rtscompare}
If $(R_{12}, R_{21},\Roner,\Rtwor)\in \Cout$, then $(R_{12}-1/2, R_{21}-1/2,\Roner-1/2,\Rtwor-1/2)\in \Rtsma.$
\end{lemma}
\begin{proof}
See proof in Appendix~\ref{append:Rtscompare_overall}.
\end{proof}

\begin{theorem}\label{thm:Rtscompare}
The EER-BR  scheme achieves any rate within half bit of the capacity  region for the two-phase protocol, i.e., if
$(R_{12}, R_{21},\Roner,\Rtwor)\in \mathcal{C}$, then $(R_{12}-1/2, R_{21}-1/2,\Roner-1/2,\Rtwor-1/2)\in \Rts.$
\end{theorem}
\begin{proof}
Comparing the  the outer bound $\overline{\mathcal{C}}$ in Theorem~\ref{thm:converse} with the achievable rate region $\Rts$ in Theorem~\ref{thm:rateRts}, and using Lemma~\ref{lem:Rtscompare},
we get
$(R_{12}, R_{21},\Roner,\Rtwor)\in \overline{\mathcal{C}} \Rightarrow (R_{12}-1/2, R_{21}-1/2,\Roner-1/2,\Rtwor-1/2)\in \Rts.$
Since $ \overline{\mathcal{C}}\supseteq \mathcal{C}$, the desired result follows.
\end{proof}

\begin{remark}
Since $\Rbc=\OBbc$, there is no loss in optimality of the EER-BR scheme for  the BC phase, as also observed for the conventional MAC scheme.
Hence, the proof of the near-optimality of the EER-BR scheme for the two-phase protocol lies mainly in Lemma~\ref{lem:Rtscompare}.
\end{remark}
\begin{remark}
A larger rate region, especially at low SNR,  is given by
$\mathsf{conv} \{ \Rts \cup \Rsep \},$ where $\mathsf{conv}$ is the convex hull operation.
This is obtained by time sharing the EER-BR scheme with the scheme based on the conventional MAC approach.
Nevertheless, the EER-BR scheme with achievable rate region $\Rts$ is sufficient to achieve near-optimality.
\end{remark}

\section{Conclusion}\label{sec:con}

We have introduced a generalized two-way relay channel, which models a three-node communication scenario where each of two nodes sends different messages to the remaining two nodes, while the third node assists.
We focused on the Gaussian setting and employs a two-phase protocol. We proposed a coding scheme based on time sharing Gaussian codes and lattice codes as well as a bit relabeling technique, which achieves within half bit of the capacity region for any channel conditions.
In a separate work \cite{ho_privatinfo_ICC10}, we have also applied the lattice coding schemes to a multi-carrier system with optimization of the time-sharing variables.

\appendix

\subsection{Proof for Theorem~\ref{thm:converse}}\label{append:proof_converse_overall}


Let $E_1, E_2, E_3$ be the error events $\{(\Whatoner, \Whattwor)\neq (\Woner, \Wtwor)\}, \{\hat{W}_{12}\neq W_{12}\}$ and $\{\hat{W}_{21}\neq W_{21}\}$, respectively. Then the error probability $\Pe$ is lower bounded  as:
$
\Pe =  \Pr\left(E_1\cup E_2 \cup E_3 \right)
\geq \max_{i=1,2,3} \{ \Pr( E_i)\}.
$
If $\Pe$ approaches zero, each  $\Pr( E_i)$ also goes to zero.
Then we have 
\be\label{eqn:fano}
\IEEEyessubnumber\label{eqn:fano1}
H(\Woner, \Wtwor|\Yrn)  \leq n\epsilon_n \\
\IEEEyessubnumber\label{eqn:fano2}
H(W_{12} |Y_2^n,  W_{21}, \Wtwor)  \leq n\epsilon_n \\
\IEEEyessubnumber\label{eqn:fano3}
H(W_{21} |Y_1^n,  W_{12}, \Woner)  \leq n\epsilon_n
\ee
where $\epsilon_n\rightarrow 0$ as $n\rightarrow \infty$. Here, \re{eqn:fano1} follows from Fano's inequality \cite{Cover06},
while \re{eqn:fano2} and \re{eqn:fano3} follow from Fano's inequality with the fact that the sources can use their previously transmitted messages as side information for decoding \cite[Lemma 2.5]{Oechtering08}.

First, we prove $R_{12} \leq  C(\Prtwo)$ if $\Pe\rightarrow 0$.  The proof for $R_{21} \leq  C(\Prone)$ is similar.
We have
\ben
n R_{12}
&\mathop{=}^{{\mybox{(a)}}}& H(W_{12}| W_{21}, \Wtwor) \\
&=& I(W_{12};Y_2^n| W_{21}, \Wtwor) + H(W_{12} |Y_2^n,  W_{21}, \Wtwor)  \\
&\mathop{\leq}^{{\mybox{(b)}}}&  I(W_{12};Y_2^n| W_{21}, \Wtwor) +n\epsilon_n \\
&\mathop{\leq}^{\mybox{(c)}}&  I( W_{21}, W_{12}, \Wtwor, \Woner ;Y_2^n) +n\epsilon_n \\
&=&  H(Y_2^n) - H(Y_2^n| W_{21}, W_{12}, \Wtwor, \Woner ) +n\epsilon_n \\
&\mathop{\leq}^{\mybox{(d)}}&  H(Y_2^n) - H(Y_2^n|X_{\rlabel}^n,  W_{21}, W_{12}, \Wtwor, \Woner ) +n\epsilon_n \\
&\mathop{=}^{\mybox{(e)}}&  H(Y_2^n) - H(Y_2^n|X_{\rlabel}^n ) +n\epsilon_n \\
&=&  I(X_{\rlabel}^n ;Y_2^n) +n\epsilon_n
\een
where (a) follows from the independence of the messages;
(b) follows from \re{eqn:fano2};
(c) follows from the chain rule of mutual information and $I(\cdot;\cdot|\cdot)\geq 0$;
(d) follows as conditioning reduces entropy;
(e) follows because $( W_{21}, W_{12}, \Wtwor, \Woner)  - X_{\rlabel}^n - Y_2^n$ forms a  Markov chain.
Note steps (d) and (e) together show that the data processing inequality holds even if side information is available to decode $\hat{W}_{12}$.
Following standard steps for the converse proof of the capacity of Gaussian channels \cite{Cover06}, we obtain  $R_{12} \leq  C(\Prtwo)$.

Next, we prove $\Roner + R_{12} \leq C(P_1 )$  if $\Pe\rightarrow 0$.  The proof for $\Rtwor + R_{21} \leq  C(P_2 )$ is similar.
We have
\ben
&& n (\Roner + R_{12}) \\
&\mathop{=}^{{\mybox{(a)}}}& H(\Woner|W_{21}, \Wtwor)+ H(W_{12}| \Woner, W_{21}, \Wtwor)
\\
&\mathop{\leq}^{{\mybox{(b)}}}&  I(\Woner;\Yrn|W_{21}, \Wtwor)+I(W_{12};Y_2^n| \Woner, W_{21}, \Wtwor) +2n\epsilon_n \\
&\mathop{\leq}^{{\mybox{(c)}}}&  I(\Woner;\Yrn|W_{21}, \Wtwor)+I(W_{12};\Yrn| \Woner, W_{21}, \Wtwor) +2n\epsilon_n \\
&\mathop{=}^{\mybox{}}&  H(\Yrn| W_{21}, \Wtwor) - H(\Yrn| \Woner, W_{12}, W_{21}, \Wtwor) +2n\epsilon_n \\
&\mathop{=}^{\mybox{(d)}}&  H(\Yrn| X^n_2, W_{21}, \Wtwor) \\
&& - H(\Yrn|  X^n_1,  X^n_2, \Woner, W_{12}, W_{21}, \Wtwor) +2n\epsilon_n \\
&\mathop{\leq}^{\mybox{(e)}}&  H(\Yrn| X^n_2) - H(\Yrn|  X^n_1,  X^n_2) +2n\epsilon_n \\
&\mathop{=}^{\mybox{}}&  I(X^n_1; \Yrn| X^n_2)  +2n\epsilon_n
\een
where (a) follows from the independence of the messages;
(b) follows from  the following inequalities
\ben
H(\Woner | \Yrn,W_{21},\Wtwor) \le H(\Woner| \Yrn) \le H(\Woner, \Wtwor| \Yrn) \\
H(W_{12} | \Yrn,\Woner, W_{21}, \Wtwor) \le  H(W_{12} | Y_2^n, W_{21}, \Wtwor)
\een
and by applying Fano's inequality \re{eqn:fano1} and \re{eqn:fano2};
(c) follows from the data processing inequality (which can be shown to hold even if $W_{12},  W_{21}, \Wtwor$ are given);
(d) follows from the fact that $X_1^n$ is a function of only $\Woner, W_{12}$ and $X_2^n$ is a function of only $\Wtwor, W_{21}$;
(e) follows from conditioning reduces entropy and because  $(\Woner, W_{12}, W_{21}, \Wtwor) -  (X^n_1,  X^n_2) - \Yrn$ forms a Markov chain.
Following standard steps for the converse proof of the capacity of Gaussian MAC channels \cite{Cover06}, we  obtain $\Roner + R_{12} \leq C(P_1 )$.

Before we prove \re{eqn:converse3}, we first prove that $\Roner + \Rtwor + R_{12} \leq  C(P_1  + P_2 )$ holds if $\Pe\rightarrow 0$.
We have
\ben
&& n (\Roner + \Rtwor + R_{12}) \\
&\mathop{=}^{{\mybox{(a)}}}& H(\Woner, \Wtwor | W_{21}) + H( W_{12}| \Woner, \Wtwor, W_{21}) \\
&\mathop{\leq}^{{\mybox{(b)}}}& I(\Woner, \Wtwor; \Yrn | W_{21}) + I( W_{12}; Y_2^n| \Woner,  \Wtwor, W_{21}) + 2\epsilon_n \\
&\mathop{\leq}^{{\mybox{(c)}}}& I(\Woner, \Wtwor; \Yrn |W_{21}) + I( W_{12}; \Yrn| \Woner, \Wtwor, W_{21}) + 2\epsilon_n \\
&\mathop{=}^{{\mybox{}}}& I(\Woner, \Wtwor, W_{12} ; \Yrn| W_{21}) + 2\epsilon_n \\
&\mathop{\leq}^{{\mybox{(d)}}}& I( \Woner, \Wtwor, W_{12}, W_{21}; \Yrn) + 2\epsilon_n \\
&\mathop{\leq}^{{\mybox{(e)}}}& I(X_1^n, X_2^n; \Yrn) + 2\epsilon_n
\een
where (a) follows from the independence of the messages;
(b) follows from conditioning reduces entropy and Fano's inequality via \re{eqn:fano1} and \re{eqn:fano2};
(c) follows from the data processing inequality (which can be shown to hold even if $\Woner, \Wtwor, W_{21}$ are given);
(d) follows from the chain rule of mutual information and from $I(\cdot;\cdot|\cdot)\geq 0$;
(e) follows from the data processing inequality.
Following standard steps for the converse proof of the capacity of Gaussian MAC channels \cite{Cover06}, we  obtain $\Roner + \Rtwor + R_{12} \leq  C(P_1  + P_2 )$  if $\Pe\rightarrow 0$.
Similarly, we can obtain $\Roner + \Rtwor + R_{21} \leq  C(P_1  + P_2 )$ if $\Pe\rightarrow 0$. Thus, \re{eqn:converse3} holds if $\Pe\rightarrow 0$.


\subsection{Proof for  Lemma~\ref{lem:Rtscompare} } \label{append:Rtscompare_overall}
%
Suppose $R_{12}\leq R_{21}$. The proof for $R_{12}\geq R_{21}$ is similar.
Without loss of generality, we let the rate tuple $(R_{12}, R_{21}, \Roner, \Rtwor)$ in $\Cout$ in Theorem~\ref{thm:converse} be $(R_0, R_0+\delta, \Roner', \Rtwor'-\delta)$, where $R_0,  \Roner', \Rtwor'\geq 0$ and $0\leq \delta\leq \Rtwor'$ (a one-to-one mapping of four variables to another four).
Since $R_0\leq C(P_1)$ from \re{eqn:converse1}, without loss of generality,  let $R_0=\alpha C(P_1)$, where $0\leq \alpha\leq 1.$
Then, $\Cout$ is alternatively given by 
\begin{IEEEeqnarray}{RRL}
\Cout=
\{ && (R_0, R_0+\delta, \Roner', \Rtwor'-\delta) : \nonumber \\
 && R_0 = \alpha C(P_1)   \nonumber \\ 
\IEEEyessubnumber\label{eqn:converse01}
&& 0\leq \Roner' \leq  (1-\alpha) C(P_1 ) \\
\IEEEyessubnumber\label{eqn:converse02}
&& 0\leq \Rtwor' \leq  C(P_2 )-\alpha C(P_1) \\
\IEEEyessubnumber\label{eqn:converse03}
&& \Roner'+\Rtwor' \leq C(P_1  + P_2 )-\alpha C(P_1) \\
&& 0\leq \alpha\leq 1, 0\leq \delta\leq \Rtwor' \;\;\;\;\;\}.
\nonumber
\end{IEEEeqnarray}
%
Fix $\alpha$ and $\delta$, where $0\leq \alpha \leq 1, 0\leq \delta\leq \Rtwor'$.
The rate $R_0=\alpha D(P_1)$ in \re{eqn:rts} differs from $R_0=\alpha C(P_1)$ in $\Cout$ by at most
$C(P_1)-D(P_1)\leq 1/2 \log(3/2)< 1/2$. To see this, recall the definition $D(x)=1/2 \log(1/2+x)$ and note that $C(P_1)-D(P_1)$ is maximized when $P_1=1/2$.
We now compare the inequalities \re{eqn:ts0}--\re{eqn:ts2} with \re{eqn:converse01}--\re{eqn:converse03}, respectively.
The first pair \re{eqn:ts0},  \re{eqn:converse01} is the same.
The second pair \re{eqn:ts1},  \re{eqn:converse02} differs in the RHS by at most
$\Gamma-C(P_1) \leq [C(P_2) - C(P_1+P_2)] + [C(2P_1)  - C(P_1)] \leq 1/2 $
since
$C(P_2)\leq C(P_1+P_2)$ and
$C(2P_1)- C(P_1)\leq 1/2$.
The last pair  \re{eqn:ts2},  \re{eqn:converse03} differs in the RHS by at most
$C(2P_1) - C(P_1) \leq 1/2$.
Thus, each achievable rate is within half bit of its respective upper bound. Since this holds for arbitrary $0\leq \alpha \leq 1, 0\leq \delta\leq \Rtwor'$, we obtain Lemma~\ref{lem:Rtscompare}.


\end{document}